\documentclass{article}
\usepackage[utf8]{inputenc}
\usepackage{amsmath}
\usepackage{amsthm}
\usepackage{bbm}
\usepackage{amsfonts}
\usepackage{natbib}
\usepackage{graphicx}
\usepackage{color}
\usepackage{comment}

\newcommand{\Bern}{\text{Bern}}

\newcommand{\Tr}{\text{Tr}}
\newcommand{\URE}{\text{URE}}

\newtheorem{theorem}{Theorem}
\newtheorem{lemma}{Lemma}
\newtheorem{assumption}{Assumption}

\newcommand{\real}{\mathbb{R}}

\newcommand{\bsX}{\boldsymbol{X}}

\newcommand{\rct}{\mathcal{R}}
\newcommand{\odb}{\mathcal{O}}

\newcommand{\err}{\varepsilon}

\newcommand{\tran}{\mathsf{T}}

\newcommand{\simiid}{\stackrel{\mathrm{iid}}\sim}

\newcommand{\dnorm}{\mathcal{N}}

\newcommand{\bstau}{\boldsymbol{\tau}}
\newcommand{\hbstau}{\boldsymbol{\hat \tau}}

\newcommand{\htaur}{\boldsymbol{\hat \tau_r}}

\newcommand{\htauo}{\boldsymbol{\hat \tau_o}}

\newcommand{\bst}{\boldsymbol{\theta}}
\newcommand{\bsxi}{\boldsymbol{\xi}}
\newcommand{\bsD}{\boldsymbol{D}}

\newcommand{\bsSig}{\boldsymbol{\Sigma}}
\newcommand{\hbsSig}{\boldsymbol{\hat \Sigma}}
\newcommand{\ident}{\boldsymbol{I}}
\newcommand{\bsdelt}{\boldsymbol{\delta}}

\newcommand{\bskap}{\boldsymbol{ \kappa}}

\title{Combining Observational and Experimental Datasets Using Shrinkage Estimators}
\author{Evan Rosenman, Guillaume Basse, Art Owen, Michael Baiocchi}

\begin{document}

\maketitle
\tableofcontents

\section{Introduction}

The modern era has yielded passive collection of massive observational datasets in areas such as e-commerce and electronic health. These data are promising and perilous. They may plausibly offer useful insights about causal effects of interest, but standard identification assumptions -- most notably, that all confounders are measured -- often fail to hold. Analysts must therefore exhibit caution before trusting causal estimates derived solely from these data. 

By contrast, a well-designed experiment will yield unbiased estimates of a causal effect, obviating the need for problematic statistical assumptions. But experimental data is  frequently expensive to obtain, and, as a consequence, often involves fewer units. Especially if one is interested in subgroup heterogeneity, this means experimental estimates can be imprecise. Hence, while observational data frequently has a ``bias problem," experimental data may suffer from a ``variance problem." 

In this paper, we consider combining data from observational and experimental sources, a problem of ``data fusion" \citep[see e.g.][]{Bareinboim7345}. In \cite{rosenman2018propensity}, we considered this problem under the assumption that all confounders were measured. This assumption -- challenging to defend in many practical problems -- ensures that all selection bias can be removed if we condition on the propensity score, the conditional probability of treatment given covariates. Practically, some bias will remain due to imperfect stratification, but it can be quantified. 

Here, we relax the assumption that all confounders are measured, meaning that residual bias of unknown magnitude can remain after stratifying on the propensity score. To derive combined estimators with desirable properties, we make use of the Stein Shrinkage literature. The classical James-Stein estimator, first introduced in \cite{stein1956inadmissibility}, considers shrinkage toward zero for a multivariate normal vector. But extensions, primarily discussed in \cite{green1991james} and \cite{green2005improved}, consider the combination of unbiased and biased estimators. 

Our contributions are threefold. First, we propose a generic procedure for deriving shrinkage estimators in this setting, making use of a generalized unbiased risk estimate. Second, we develop two new estimators, prove finite sample conditions under which they have lower risk than an estimator using only experimental data, and show that each achieves a notion of asymptotic optimality. Third, we draw connections between our approach and state-of-the-art results in sensitivity analysis, including proposing a method for evaluating the feasibility of our estimators. 

The remainder of this paper proceeds as follows. In Section \ref{sec:relatedLit}, we review literature on the data fusion problem and Stein Shrinkage. In Section \ref{sec:notationAndAssumptions}, we introduce notation and assumptions. In Section \ref{sec:proposedEstimators}, we develop our procedure and introduce our estimators. In Section \ref{sec:sens}, we discuss sensitivity analysis, and in Section \ref{sec:dataAnalysis}, we demonstrate the utility of our methods on a simulated dataset. The Appendix contains some of our longer proofs. Section \ref{sec:conclusions} concludes. 

\section{Related Literature}\label{sec:relatedLit}

Variants of the data fusion problem have a long history in the literature. In the middle of the twentieth century, \cite{campbell1957factors} introduced the concepts of ``internal validity" and ``external validity" to distinguish between challenges of treatment effect estimation and generalizability in quasi-experimental research. This paradigm was widely adopted among social scientists. The problem of extending causal findings across different domains is now known under the broader banner of ``transportability," which subsumes results from the meta-analysis and treatment effect heterogeneity literatures \citep{Bareinboim7345}. In this context, observational data is often used to examine whether causal effects from an experiment can be generalized to a target population \citep{cole2010generalizing, stuart2011use}. \cite{hartman2015sate} derived assumptions and placebo tests for identifying such population treatment effects from RCTs. 

There has been comparatively less work on incorporating causal effects computed using observational data, likely owing to concerns about introducing bias into the estimation. One approach is to assume unconfoundedness in the observational study, meaning that all variables affecting the treatment assignment and the outcome are measured. This is our approach in \cite{rosenman2018propensity}, and is also used in \cite{athey2019surrogate}. Some prior work has attempted to weaken this assumption, but typically introduces alternative assumptions. In \cite{kallus2018removing}, the authors assume that the hidden confounding has a parametric structure that can be modeled effectively. In \cite{peysakhovich2016combining}, it is assumed the bias preserves unit-level relative rank ordering.

Though they were not focused on questions of causality, \cite{green1991james} addressed the question of combining biased and unbiased estimators in the Empirical Bayes framework. They suppose they have two $K$-dimensional multivariate normal vectors $\htaur$ and $\htauo$ such that $\htaur$ has mean $\bst$ and $\htauo$ has mean $\bst - \bsxi$. The vectors are assumed homoscedastic with covariance matrices $\bsSig_r = \sigma^2 \ident_K$ and $\bsSig_o = v^2 \ident_K$. The goal is to estimate $\bst$ under the $L_2$ loss. The authors propose the estimator
\[ \htauo + \left( 1 - \frac{(K-2)\sigma^2}{||\htauo - \htaur ||^2} \right)_+ (\htaur - \htauo) \]
and show that it dominates $\htaur$ in terms of risk. Unsurprisingly, if $||\bsxi||^2$ is very small, the estimator underperforms a simple precision-weighted estimator. Yet, unlike the precision-weighted estimator, the proposed estimator has bounded risk as the biases grow. 

A key question is how to generalize these results to the heteroscedastic case. The follow-up, \cite{green2005improved}, proposes two estimators designed for this case. The first, 
\[ \bsdelt_1 = \htauo + \left( 1 - \frac{a}{(\htaur - \htaur)^\tran  \hbsSig_r^{-1} \left( \htaur - \htauo \right)}\right) \left(\htaur - \htauo \right) \] 
can be shown to dominate $\htaur$ (if $\hbsSig_r$ is perfectly estimated) under precision-weighted squared-error loss, (i.e. the squared coordinate residuals are scaled by the corresponding $\frac{1}{\sigma_{rk}^2}$ term). Under conventional squared error loss, they instead propose 
\[ \bsdelt_2 = \htauo + \left(\ident_K - \frac{a \hbsSig_r^{-1}}{(\htaur - \htauo)^\tran  \hbsSig_r^{-2} (\htaur - \htauo)} \right)(\htaur - \htauo ) \] 
The shrinkage parameter $a$ is optimized at $K-2$ for $\bsdelt_2$, while it depends on the value of $\bsxi$ for $\bsdelt_1$. Absent information about $\bsxi$, however, the authors default to using $a = K-2$ for this estimator as well.

\section{Notation, Assumptions, and Set-Up}\label{sec:notationAndAssumptions}

\subsection{Setup}\label{subsec:setup}

Suppose we have access to an observational study with units $i$ in indexing set $\mathcal{O}$ such that $|\mathcal{O}| = n_o$. We also have access to an RCT with $i \in \mathcal{R}$ and $|\mathcal{R}| = n_r$. We associate with each unit $i \in \mathcal{O} \cup \mathcal{R}$, a set of constants: 
\begin{itemize}
    \item Each unit has a pair of fixed, unseen potential outcomes $(Y_i(1), Y_i(0))$. These represent the unit's value for an outcome of interest in the presence or absence of treatment, respectively. 
    \item We measure a covariate vector $\bsX_i \in \mathbb{R}^p$, for each unit $i$. 
    \item Each unit also has a value for an \emph{unmeasured} covariate, denoted $U_i$. 
\end{itemize}
For units $i \in \mathcal{O}$, we also associate a propensity score, $p_i \in [0, 1]$, denoting the probability that the unit receives treatment. The propensity score is unknown to the researcher. 

We also associate with each $i \in \mathcal{O} \cup \mathcal{R}$ two random quantities: 
\begin{itemize}
    \item Each unit has a treatment indicator, $W_i$, where $W_i = 1$ indicates that the unit receives treatment and $W_i = 0$ indicates that the unit is untreated. 
    \item The value of $W_i$ defines the observed outcome, which is given by: 
    \[ Y_i = W_i Y_i(1) + (1-W_i) Y_i(0)\]
\end{itemize}

\subsection{Assumptions and Loss Function}

We suppose a stratification scheme is known, such that there are $k = 1, \dots, K$ strata and each has an associated population weight $w_1, \dots, w_K$. We treat the strata as non-random. We define indexing subsets $\mathcal{O}_k, \mathcal{R}_k$ (with cardinalities $n_{ok}, n_{rk}$) to identify units in each stratum. We make simple assumptions about the allocation to treatment in the two studies. 

\begin{assumption}[Allocations to Treatment]
For $i \in \mathcal{O}$, $W_i \sim \Bern(p_i)$ for $p_i = f(X_i, U_i)$, a function of the observed and unobserved covariates. For $i \in \mathcal{R}$, treatment is allocated via a simple random sample of size $n_{rkt}$ for $k = 1, \dots, K$.
\end{assumption}

Key to our analysis is the additional assumption: 
\begin{assumption}[Common Treatment Effect]\label{ass:cte}
The average causal effects are identical between the two populations i.e. for all $k = 1, \dots, K$: 
\[ \tau_k = \frac{1}{n_{ok}}\sum_{i \in \mathcal{O}_k} Y_i(1) - Y_i(0) = \frac{1}{n_{rk}} \sum_{i \in \mathcal{R}_k} Y_i(1) - Y_i(0) \] 
\end{assumption}
\noindent Denote a target of estimation $\bstau = \left(\tau_1, \dots, \tau_K \right)$. 

Assumption \ref{ass:cte} may be more or less plausible based on our experimental set-up. It may be the case, for example, that the experiment is a multi-arm trial involving $K$ different potential treatments. Within stratum $k$, all of the units would either receive treatment option $k$ or would be assigned to a control condition. This is the setting of \cite{dimmery2019shrinkage}. In this case, the observational dataset would be an agglomeration of observational datasets in which the analogous treatments were available to units. Assumption \ref{ass:cte} then becomes an assumption that each potential treatment has constant treatment effect such that the effects are transportable across the two datasets. 

Alternatively, we may assume there is only a single treatment but substantial heterogeneity across different subgroups. These subgroups may be known a priori, or they may be discovered by deploying a modern method used for heterogeneous treatment effect estimation \citep{wager2018estimation, hill2011bayesian}. The subgroups will then be defined by covariates. Assumption \ref{ass:cte} would thus depend on correct identification of the relevant subgroups. 

In either case, the assumption is a mathematical convenience, allowing us to define a shared target of estimation. It is unlikely to hold precisely in practice, as we are working with finite sample averages rather than population means. We might, alternatively, assume that the average treatment effects differ by no more than a factor of $O(1/n)$, or that they are equivalent only after some statistical adjustment. We do not explore such possibilities here, but consider them for future work. 

Under Assumption \ref{ass:cte}, we consider our aggregate loss. We are interested in the individual causal effects within each stratum $k$, rather than an overall ATE. In full generality, we define our loss function as: 
\[ \mathcal{L}(\bstau, \hbstau) =  \frac{1}{K} \sum_k d_k \left( \hat \tau_k - \tau_k \right)^2 \hspace{3mm} \text{ where } \hspace{3mm} d_k > 0, \hspace{1mm} \sum_k d_k = 1\] 
The stratum weights $d_k$  correspond to how much we ``care" about accuracy in that particular stratum. Typically, we would want $d_k \approx w_k$, where $w_k$ is the population weight of stratum $k$ for a target population of interest. Lacking this, we can instead use the observational data to define a surrogate weight
\[ d_k = \frac{n_{ok}}{n_o} \,.\] 
In other words, we use the observational dataset frequencies to estimate the population frequencies of the strata, and then weight the strata based on these estimated population frequencies. We denote as $\bsD$ the diagonal matrix whose entries are given by the $d_k/K$, such that 
\[ \mathcal{L}(\hbstau, \bstau) =  \left(\hbstau - \bstau \right)^\tran \bsD \left(\hbstau - \bstau \right) \,.\] 

\subsection{Estimator Distributions}
We define the following estimators: 
\begin{align*}
    \hat \tau_{ok} &= \frac{\sum_{i \in \mathcal{O}_k} W_i Y_i}{\sum_{i \in \mathcal{O}_k} W_i } - \frac{\sum_{i \in \mathcal{O}_k} (1-W_i)Y_i}{\sum_{i \in \mathcal{O}_k} (1-W_i)} \\
    \hat \tau_{rk} &= \frac{\sum_{i \in \mathcal{R}_k} W_i Y_i}{n_{rkt}} - \frac{\sum_{i \in \mathcal{R}_k} (1-W_i)Y_i}{n_{rkc}} 
\end{align*}
where $n_{rkc} = n_{rk} - n_{rkt}$. Denote $\htauo = \left( \hat \tau_{o1}, \dots, \hat \tau_{oK} \right)$ and $\htaur$ analogously.

Per the discussion in Section \ref{subsec:setup}, we operate in the randomization framework, meaning that potential outcomes are fixed and the only random quantity is the treatment assignment. We assume sufficient sample sizes and regularity conditions such that a Central Limit Theorem holds for $\htaur$. For more details on the technical conditions for this result, see \cite{li2017general}. Hence, we have approximately $\htaur \sim \dnorm \left(\bstau, \bsSig_r \right)$. 

We need not make assumptions about the distribution about $\htauo$, though we denote its mean as $\bstau + \bsxi$, where $\bsxi$ represents a $K$-dimensional bias parameter. The covariance matrix is denoted $\bsSig_o$. The bias results from correlation between the potential outcomes $(Y_i(1), Y_i(0))$ and the propensity scores $p_i$ within each stratum $k$. Denote these stratum-specific correlations as $s_{tk}$ and $s_{ck}$ and the average propensity scores within each stratum as $\bar p_k$. Then we can use the Delta Method to observe
\[ \xi_k = \frac{s_{tk}}{\bar p_k} + \frac{s_{ck}}{1 - \bar p_k} + O\left( \frac{1}{n_{ok}}\right) \,.\] 
A full derivation can be found in \cite{rosenman2018propensity}. 

Our assumptions imply that $\bsSig_o$ and $\bsSig_r$ will be diagonal matrices. We denote the diagonal entries of $\bsSig_o$ as $\sigma_{o1}^2, \dots, \sigma_{oK}^2$ with analogous definitions for $\bsSig_r$. 

\section{Proposed Estimators}\label{sec:proposedEstimators}

\subsection{Preliminaries}

We begin with a mild generalization of a result from \cite{strawderman2003minimax}. 

\begin{theorem}[Estimator Risk]\label{thm:estRisk}
Suppose we have $\boldsymbol{Z} \sim \dnorm(\bst, \bsSig)$, random $\boldsymbol{Y}$, and $\mathcal{L}(\bst, \boldsymbol v) = (\boldsymbol v-\bst)^\tran \boldsymbol D(\boldsymbol v-\bst)$ where $\bsSig = \text{diag}(\sigma_1^2,, \dots, \sigma_{k}^2)$ and $\bsD = 1/K\cdot\text{diag}(d_1, \dots, d_K)$ is a diagonal weight matrix quantifying the relative importance of the $K$ components. Then for 
\[ \bskap(\boldsymbol Z, \boldsymbol Y) = \boldsymbol Z + \boldsymbol \Sigma \boldsymbol g(\boldsymbol Z, \boldsymbol Y) \] 
where $\boldsymbol g(\boldsymbol Z, \boldsymbol Y) $ is a function of $\boldsymbol Z$ and $\boldsymbol Y$ that is differentiable, satisfying $E(||\boldsymbol g||^2) < \infty$, we have 
\[ R(\bst, \bskap(\boldsymbol Z, \boldsymbol Y)) = E \left( \mathcal{L}(\bst, \bskap(\boldsymbol Z, \boldsymbol Y)) \right) =  \frac{1}{K}\left( \Tr\left( \bsSig \bsD \right) + E \left(\sum_{k = 1}^K \sigma_{k}^4 d_k  \left(g_k^2(\boldsymbol Z, \boldsymbol Y) + 2 \frac{\partial g_k(\boldsymbol Z, \boldsymbol Y)}{\partial Z_k} \right) \right)\right)\,. \] 
\end{theorem}
\begin{proof}
Fix a vector $\boldsymbol y$ and define $\boldsymbol g^{(\boldsymbol y)}(\boldsymbol Z) = g(\boldsymbol Z, \boldsymbol y)$ and $\bskap^{(\boldsymbol y)}(\boldsymbol Z) = \boldsymbol Z + \boldsymbol \Sigma \boldsymbol g^{(\boldsymbol y)}(\boldsymbol Z).$ Observe $\boldsymbol g^{(\boldsymbol y)}(\boldsymbol Z)$ is a differentiable function of $\boldsymbol Z$, $E(||\boldsymbol g^{(y)}||^2) < \infty$. By Theorem 3.1 in \cite{strawderman2003minimax}, we thus have 
\[ R(\bst, \bskap^{(\boldsymbol y)}( \boldsymbol Z))) = \frac{1}{K}\left(\Tr\left( \bsSig \bsD \right) + E \left(\sum_{k = 1}^K \sigma_{k}^4 d_k  \left(\left(g^{(\boldsymbol y)}_k( \boldsymbol Z)\right)^2 + 2 \frac{\partial g^{(\boldsymbol y)}_k(\boldsymbol Z)}{\partial Z_k} \right) \right)\right)\,. \] 
By the Tower Rule, we know
\begin{align*}
E(R(\bst, \bskap(\boldsymbol Z, \boldsymbol Y))) &= E(E(R(\bst, \bskap(\boldsymbol Z, \boldsymbol Y)) \mid \boldsymbol Y =  \boldsymbol y)) \\
&= E(E(R(\bst, \bskap(\boldsymbol Z, \boldsymbol y)) \mid \boldsymbol Y = \boldsymbol y)) \\
&= E \left(R(\bst, \bskap^{(\boldsymbol y)}(\boldsymbol Z))\right)\,,
\end{align*}
and the result follows.
\end{proof}

From Theorem \ref{thm:estRisk}, we can obtain a generalization of Stein's Unbiased Risk Estimate \citep{stein1981estimation} for our setting,
\[ \URE(\bst, \bskap(\boldsymbol Z, \boldsymbol Y)) = \frac{1}{K}\left( \Tr\left( \bsSig \bsD \right) + \sum_{k = 1}^K \sigma_{k}^4 d_k  \left(g_k^2(\boldsymbol Z, \boldsymbol Y) + 2 \frac{\partial \boldsymbol g_k(\boldsymbol Z, \boldsymbol Y)}{\partial Z_k} \right) \right)\,. \] 

Our procedure for deriving estimators will be based on this unbiased risk estimate. For each, we will follow these steps: 
\begin{enumerate}
\item Posit a structure for the shrinkage estimator 
\item Derive a functional form for the shrinkage factor by optimizing URE, assuming the shrinkage factors are known a priori. This will be our ``base" estimator. 
\item (Optional) Generate a ``corrected" version of the estimator that attempts to account for the fact that the shrinkage factors are estimated from the data. 
\end{enumerate}

\subsection{$\bskap_1$, Common Shrinkage Factor}

We consider shrinkage estimators which share a common shrinkage factor across components. Denote a generic estimator as 
\[ \bskap(\lambda, \htaur, \htauo) = \htaur - \lambda(\htaur - \htauo)\,,\]
where $\lambda$ is our common shrinkage factor. 
%We assume $\lambda$ lies in the unit interval, such that if $\lambda$ equals 0, the estimator resolves to $\htaur$, and if it equals 1, the the estimator resolves to $\htauo$. 

For our second step, we will select $\lambda$ by minimizing the unbiased risk estimate. This approach has substantial precedent in the literature \citep[see e.g.][]{li1985stein, xie2012sure}. Supposing $\lambda$ is fixed ahead of time, the unbiased risk estimate is
\begin{equation}\label{eq:URE.fixed.lambda}
\URE(\bstau, \bskap(\lambda, \htaur, \htauo)) =  \Tr\left( \bsSig_r \bsD \right) + \lambda^2 \left( \htauo - \htaur\right)^\tran  \bsD\left( \htauo - \htaur\right) - 2 \lambda \Tr(\bsSig_r \bsD)\,.
\end{equation} 
This expression is strictly convex in $\lambda$ as long as $\htaur \neq \htauo$. We seek to find
\[ \lambda_1^{\URE} = \min_{\lambda} \URE(\bstau, \bskap(\lambda, \htaur, \htauo)) \,.\] 
Simple calculus tells us the unbiased risk estimate achieves its minimum at%can either achieve its minimum in the unit interval, in which case 
\[ \lambda_1^{\URE} =\frac{\Tr(\bsSig_r \bsD)}{\left( \htauo - \htaur\right)^\tran \bsD\left( \htauo - \htaur\right) } \,,\]
giving us the estimator
%or else it must be strictly decreasing over the unit interval and thus achieve its minimum allowable value at $\lambda = 1$. Combining these cases and rearranging yields
\[ \bskap_1 = \bskap(\lambda_1^{\URE}, \htaur, \htauo) = \htaur -  \frac{\Tr(\bsSig_r \bsD)}{(\htauo - \htaur)^\tran \bsD ( \htauo - \htaur) } \left( \htaur - \htauo \right) \,.\]

This estimator generalizes the estimator of \cite{green1991james} to the heteroscedastic, weighted-loss case, and the interpretation is similar. The oracle weighted-MSE-optimal shrinkage factor is 
\[ \lambda_{\text{opt}} = \frac{\Tr(\bsSig_r\bsD)}{\Tr(\bsSig_r\bsD) + \Tr(\bsSig_o \bsD) + \bsxi^\tran \bsD^2 \bsxi} \,.\] 
The denominator cannot be estimated from the data because the bias is unknown. But we observe that the denominator is precisely the expectation of $\left( \htauo - \htaur\right)^\tran \bsD\left( \htauo - \htaur\right)$, and so we substitute this value as our ``best guess." 

The following lemma gives us a testable condition under which $\bskap_1$ is strictly better than $\htaur$ in terms of risk.  
\begin{lemma}\label{lemma:delta1DominatesTauR}
Suppose $4 \max_k d_k \sigma_{rk}^2 \leq \sum_k d_k \sigma_{rk}^2$. Then $\bskap_1$ dominates $\htaur$ under our loss function. 
\end{lemma} 
\begin{proof}
Applying Theorem  \ref{thm:estRisk}, we have
\begin{align*}
R\left(\bstau, \htaur \right) &- R \left( \bstau, \bskap(\lambda_1^{\URE}, \htaur, \htauo)  \right) = E \left( \URE(\bstau, \htaur) \right) - E \left( \URE(\bstau, \bskap(\lambda_1^{\URE}, \htaur, \htauo) ) \right) \\
&= E \left(  -\frac{\Tr(\bsSig_r \bsD)^2}{\left( \htauo - \htaur\right)^\tran \bsD\left( \htauo - \htaur\right)} + \frac{4 \Tr(\bsSig_r \bsD) \left( \htauo - \htaur\right)^\tran \bsD^2 \bsSig_r\left( \htauo - \htaur\right)}{\left(\left( \htauo - \htaur\right)^\tran \bsD\left( \htauo - \htaur\right)\right)^2} \right) \\
& \leq \Tr(\bsSig_r \bsD)  E \left(  -\frac{\Tr(\bsSig_r \bsD)}{\left( \htauo - \htaur\right)^\tran \bsD\left( \htauo - \htaur\right)} + \frac{4 \left(\max_k d_k \sigma_{rk}^2\right) \left( \left( \htauo - \htaur\right)^\tran \bsD \left( \htauo - \htaur\right)\right)}{\left(\left( \htauo - \htaur\right)^\tran \bsD\left( \htauo - \htaur\right)\right)^2} \right) \\
&= \Tr(\bsSig_r \bsD) E \left( \frac{4 \max_k \sigma_{rk}^2 - \Tr(\bsSig_r \bsD) }{\left( \htauo - \htaur\right)^\tran \bsD\left( \htauo - \htaur\right)}\right) 
\end{align*}
Under our condition, the numerator is nonpositive, and hence the risk difference is nonpositive. 
\end{proof}

Note that the condition used in Lemma \ref{lemma:delta1DominatesTauR} requires that our dimension be at least four in order to guarantee a reduction in risk. Hence, the required dimension is at least as large as that required for risk reduction when using the classical James-Stein estimator to shrink homoscedastic estimates toward their grand mean \citep{efron2012large}. In our setting, it means there must be a minimum of four strata -- and possibly more, if the variances and weights vary substantially across strata. In the multi-arm trial case, this means at least four distinct treatments; in the heterogeneous treatments effects case, this means at least four different subgroups for whom we believe causal effects differ. 

We now consider some improvements to our estimators. We can restrict our shrinkage factor to lie between 0 and 1, an improvement also applied in \cite{green1991james} and \cite{green2005improved} and based on results in \cite{baranchik1964multiple}. Some reorganization allows us to write the estimator as 
\[ \bskap_{1+} = \htauo + \left( 1 - \lambda_1^{\URE} \right)_+ \left( \htaur - \htauo \right) \,.\] 

This estimator possesses the following desirable property. 
\begin{theorem}[$\bskap_{1+}$ Asymptotic Risk]\label{thm:delta1asymptotics}
Under mild conditions, in the limit $K \to \infty$, $\bskap_{1+}$ has the lowest risk among all estimators with a shared shrinkage factor across components. 
\end{theorem}
\begin{proof}
See the first proof in the Appendix. 
\end{proof}

In our optional third step, we consider applying a correction factor to the estimator. As motivation, observe that the risk of $\bskap_{1}$is not obtained via the expectation of (\ref{eq:URE.fixed.lambda}) evaluated at $\lambda_1^{\URE}$. This is because $\lambda_1^{\URE}$ is not actually known a priori, but rather it is estimated from the data; hence, we pay an additional risk penalty. Accounting for this additional penalty, we observe that it is preferable to shrink by less than $\lambda_1^{\URE}$. 

We can modify our estimator by optimizing a scaling value $a$ applied to our shrinkage factor. Note that we could improve this process ad infinitum -- estimating correction factors from the data, and then seeking to correct for the penalty induced by using the data to estimate the correction factor. We choose to terminate at one iteration and compare performance in simulations and data analyses to follow. 

We observe
\begin{align*}
\URE(\bstau, \bskap(a\cdot\lambda_1^{\URE}, \htaur, \htauo) ) &= \Tr\left( \bsSig_r \bsD \right) + \frac{(a^2 - 2 a) \Tr(\bsSig_r \bsD)^2}{\left( \htauo - \htaur\right)^\tran \bsD\left( \htauo - \htaur\right)} + \\&\hspace{5mm} \frac{4a \Tr(\bsSig_r \bsD) \left( \htauo - \htaur\right)^\tran \bsD^2 \bsSig_r\left( \htauo - \htaur\right)}{\left(\left( \htauo - \htaur\right)^\tran \bsD\left( \htauo - \htaur\right)\right)^2}\,,
\end{align*}
and optimizing over $a$ yields
\[ a_1^{\star} = 1 - \frac{2\left( \htauo - \htaur\right)^\tran \bsD^2 \bsSig_r\left( \htauo - \htaur\right) }{\left( \htauo - \htaur\right)^\tran \bsD\left( \htauo - \htaur\right)}\cdot \frac{1}{\Tr(\bsSig_r \bsD)}\,. \] 
This yields the modified estimator
%\[ \bskap_{1}^{\star} = \htauo + \left( 1 - a_1^{\star} \lambda_1^{\URE} \right)_+ \left( \htaur - \htauo \right) \,.\]
\[ \bskap_1^{\star}  = \htaur -  a_1^{\star}\frac{\Tr(\bsSig_r \bsD)}{(\htauo - \htaur)^\tran \bsD ( \htauo - \htaur) } \left( \htaur - \htauo \right) \,.\]

This ``corrected" estimator will not necessarily outperform $\bskap_{1}$ or $\bskap_{1+}$, because yet more factors are being estimated from the data. However, it will provably outperform $\bskap_{1}$ if the stratum variances and weights are sufficiently concentrated. 

\begin{lemma}\label{lemma:delta1Superiority}
$\bskap_1^{\star}$ has risk no greater than $\bskap_1$ if 
\[ \max_k d_k^2 \sigma_{rk}^4 \leq \frac{3}{2} \left(\min_k \sigma_{rk}^2 d_k\right)^2 \,.\] 
\end{lemma}
\begin{proof}
Applying Theorem  \ref{thm:estRisk}, we have
\begin{align*}
R\left(\bstau, \bskap_1^{\star} \right) - R\left(\bstau, \bskap_1 \right) &= E \left( \URE(\bstau, \bskap_1^{\star}) \right) - E \left( \URE(\bstau, \bskap_1) \right)  \\
&= -E \left( 8\frac{\left( \left( \htauo - \htaur\right)^\tran \bsD^3 \bsSig_r^2 \left( \htauo - \htaur \right)\right)^2}{\left( \left( \htauo - \htaur\right)^\tran \bsD \left( \htauo - \htaur \right)\right)^3} - 12\frac{\left( \left( \htauo - \htaur\right)^\tran \bsD^2 \bsSig_r \left( \htauo - \htaur \right)\right)^2}{\left( \left( \htauo - \htaur\right)^\tran \bsD \left( \htauo - \htaur \right)\right)^3} \right) \\
& \leq  \left( 8 \max_k \sigma_{rk}^4 d_k^2 - 12 \left( \min_k \sigma_{rk}^2 d_k \right)^2 \right)E \left( \frac{1}{\left( \left( \htauo - \htaur\right)^\tran \bsD \left( \htauo - \htaur \right)\right)} \right)
\end{align*}
This term is negative as long as 
\[ 8\max_k d_k^2 \sigma_{rk}^4 \leq 12 \left(\min_k \sigma_{rk}^2 d_k\right)^2 \,,\]
which simplifies to our given condition.
\end{proof}

\subsection{$\bskap_2$, Variance-Weighted Shrinkage Factors}
We may instead want to choose shrinkage factors on a component-by-component basis. One heuristic is variance-weighted shrinkage: we shrink components by a factor proportional to variance. The estimator thus relies more heavily on the RCT estimate for entries $k$ for which $\sigma_{rk}^2$ is small, and more heavily on the observational estimate for entries $k$ for which $\sigma_{rk}^2$ is large. 

A generic estimator takes the form 
\[ \bskap(\lambda \bsSig_r, \htaur, \htauo)  = \htaur - \lambda \bsSig_r(\htaur - \htauo) \,.\]

We can follow the same procedure as in the prior section: minimize the unbiased risk estimate to determine the functional form of the estimator. We find that%; compute the additional penalty incurred by estimating $\lambda$ from the data; and modify the estimator. 
\[ \lambda_2^{\URE} = \min_{\lambda} \URE(\bstau, \bskap(\lambda \bsSig_r, \htaur, \htauo) ) = \frac{\Tr(\bsSig_r^2 \bsD)}{(\htauo - \htaur)^\tran  \bsSig_r^2 \bsD (\htauo - \htaur)}\,, \] 
yielding the estimator
\[ \bskap_2 = \bskap(\lambda_2^{\URE}, \htaur, \htauo) = \htaur -  \frac{\Tr(\bsSig_r^2 \bsD)\bsSig_r}{(\htauo - \htaur)^\tran  \bsSig_r^2 \bsD (\htauo - \htaur) } \left( \htaur - \htauo \right) \,.\]
and its positive-part analogue
\[ \bskap_{2+} = \htauo + \left( \ident_K -  \frac{\Tr(\bsSig_r^2 \bsD) \bsSig_r}{(\htauo - \htaur)^\tran  \bsSig^2 \bsD (\htauo - \htaur)} \right)_+ \left( \htaur - \htauo \right) \,.\] 

These estimators have analogous finite sample and asymptotic properties to those described in the prior section. They are described in the Lemma and Theorem that follow. 
\begin{lemma}\label{lemma:delta2DominatesTauR}
Suppose $4 \max_k d_k^2 \sigma_{rk}^4 \leq \sum_k d_k^2 \sigma_{rk}^4$. Then $\bskap_2$ dominates $\htaur$ under our loss function. 
\end{lemma} 
\begin{proof}
The result follows from the same argument used in the proof of Lemma \ref{lemma:delta2DominatesTauR}.
\end{proof}

\begin{theorem}[$\bskap_2$ Asymptotic Risk]\label{thm:delta2asymptotics}
Under mild conditions, in the limit $K \to \infty$, $\bskap_2$ has the lowest risk among all estimators with a variance-weighted shrinkage factor across components. 
\end{theorem}
\begin{proof}
See the second proof in the Appendix. 
\end{proof}

We can also apply the same method of estimating a scaling correction from the data. The scaling value is given by 
\[ a_2^{\star} = 1 - \frac{2(\htauo - \htaur)^\tran  \bsSig_r^4 \bsD^2 (\htauo - \htaur)}{(\htauo - \htaur)^\tran  \bsSig_r^2 \bsD (\htauo - \htaur)} \cdot \frac{1}{\Tr(\bsSig_r \bsD)}\,. \] 
which gives us the estimator 
\[ \bskap_{2+}^{\star} = \htauo + \left( \ident_K -  \frac{a_2^{\star}\Tr(\bsSig_r^2 \bsD) \bsSig_r}{(\htauo - \htaur)^\tran  \bsSig^2 \bsD (\htauo - \htaur)} \right)_+ \left( \htaur - \htauo \right) \,.\] 

\subsection{Practical Considerations}\label{sec:pracCon}

\subsubsection{Variance Estimation}

In practice, $\bsSig_r$ will not be known. As in \cite{green2005improved}, we suggest replacing it with an estimate, $\hbsSig_r$. Under Assumption \ref{ass:cte}, the estimator 
\[ \hat \sigma_{rk}^2 = \frac{1}{n_{rkt}} \sum_{i \in \mathcal{R}_k} W_i \left( Y_i - \bar{Y}_{rkt}\right)^2 + \frac{1}{n_{rkc}} \sum_{i \in \mathcal{R}_k} (1-W_i) \left( Y_i - \bar{Y}_{rkc}\right)^2 \]
where 
\[ \bar{Y}_{rkt} = \frac{1}{n_{rkt}} \sum_{i \in \mathcal{R}_k} W_i Y_i \hspace{3mm} \text{ and } \hspace{3mm} \bar{Y}_{rkc} = \frac{1}{n_{rkc}} \sum_{i \in \mathcal{R}_k} (1-W_i) Y_i\]
is unbiased for $\sigma_{rk}^2$. 

However, if there is heterogeneity in the treatment effect within strata, then $\hat \sigma_{rk}^2$ will be a biased estimator of $\sigma_{rk}^2$, and will tend to overestimate the variance \citep{Imbens:2015:CIS:2764565}. If we are using $\bskap_{1+}$ this bias will, in expectation, translate to \emph{more} shrinkage toward $\htauo$ from $\htaur$, because the shrinkage factor is linearly proportional to $\Tr(\bsSig_r \bsD)$. If we are using $\bskap_{2+}$, then we may over-shrink some components and under-shrink others. 

There are several possible ways to mitigate this issue. One is to choose smaller strata such that Assumption \ref{ass:cte} is likelier to hold. Another is to consider a variety of possible correlations between the potential outcomes in each stratum (where a correlation of 1 corresponds to Assumption \ref{ass:cte} being true, and lower values correspond to more heterogeneity in treatment effect), and then compute a ``menu" of possible shrinkage estimators based on the associated variance estimates. For details on computing the variance estimates under a choice of potential outcomes correlation, see Chapter 6 of \cite{Imbens:2015:CIS:2764565}. 

\subsubsection{Propensity Score Adjustment}\label{subsubsec:psa}

Because treatment is not randomized in the observational study, there will be selection bias. We do not assume unconfoundedness, but assume that some relevant covariates are measured. Hence, we can reduce (but not eliminate) bias by making use of the estimated propensity score. Because the observational study is assumed to be much larger than the RCT, adjusting by the estimated propensity score will often be good practice: any increase in variance may be compensated by a decrease in bias. 

Estimation of the propensity score will depend on the problem set-up. If the strata $k$ represent different treatments, then a different propensity model should be fit in each arm. If they represent subgroups with different treatment effects, then a single propensity model can be fit. In the former case, we will obtain a propensity score $\hat p_i = f_k(X_i)$ for each unit $i \in \mathcal{O}_k$, where $f_k(\cdot)$ may represent a logistic regression or other binary classification model. In the latter case, $\hat p_i = f(X_i)$ for each unit $i \in \mathcal{O}$. 

There are many ways in which to adjust for the propensity score in order to reduce bias, such as matching, stratification, and regression \citep[see e.g.][]{imbens2015causal}. We advocate stabilized inverse probability weighted (SIPW) estimation, where 
\[ \hat \tau_{ok} = \sum_{i \in \mathcal{O}_k} \frac{W_i Y_i}{\hat p_i} \left( \sum_{i \in \mathcal{O}_k} \frac{W_i}{\hat p_i}\right)^{-1} - \sum_{i \in \mathcal{O}_k} \frac{(1-W_i) Y_i}{1-\hat p_i} \left( \sum_{i \in \mathcal{O}_k} \frac{1-W_i}{1-\hat p_i}\right)^{-1}\,. \] 
This is simply the Horvitz-Thompson inverse probability weighted estimator with normalized weights. As we will see in the next section, the SIPW method will admit a relatively straightforward sensitivity analysis, allowing analysts to better quantify the amount of bias implied by the shrinkage estimator.   

\section{Sensitivity Analysis}\label{sec:sens}

In this section, we will consider sensitivity analysis when using $\lambda_{1+}$ with the estimation strategy described in Section \ref{subsubsec:psa}. 

\subsection{Set-Up}

Recall our interpretation of $\bskap_{1+}$ as estimating the weighted-MSE-optimal tradeoff factor $\lambda_{\text{opt}}$ from the data, where 
\[ \lambda_{\text{opt}} = \frac{\Tr(\bsSig_r\bsD)}{\Tr(\bsSig_r\bsD) + \Tr(\bsSig_o \bsD) + \bsxi^\tran \bsD^2 \bsxi} \,.\] 
The numerator $\Tr(\bsSig_r \bsD)$ is directly estimable, while we use the weighted norm of the discrepancy between $\htauo$ and $\htaur$ to estimate the denominator. 

Sensitivity analysis provides us an alternate approach to estimating $\lambda_{\text{opt}}$. We can posit a model for the level of violation of unconfoundedness in each stratum $k$; compute the worst-case bias and variance under this model; and plug these ``maximally pessimistic" estimates into the above formula. Such an approach would not make use of the parallel estimates of the causal effects to estimate the shrinkage factor. Rather, it would translate a set of untestable assumptions about the level of confounding into a conservative estimation strategy for trading off between $\htauo$ and $\htaur$.

This approach is straightforward in the case when SIPW estimation is used in the observational study, owing to recent work by \cite{zhao2019sensitivity}. The authors propose a marginal sensitivity model that extends the widely-used Rosenbaum sensitivity model \citep{rosenbaum1987sensitivity}. Crucially, this allows the degree of the confounding to be summarized by a single value, $\Gamma$, which bounds the odds ratio of the true treatment probability and the estimated treatment probability for all units in the observational study. $\Gamma = 1$ implies no unmeasured confounding, while larger values of $\Gamma$ imply greater deviations from this assumption. This quantity is very similar to the one used in Rosenbaum's work, lending it ready interpretability for researchers familiar with the Rosenbaum sensitivity model.

Under a given choice of $\Gamma$, Zhao and co-authors seek to derive valid confidence intervals for inverse probability weighting estimators of causal effects. They show that the worst-case bias under $\Gamma$ can be determined through a linear fractional programming problem. The variance contribution to the confidence intervals is estimated via a bootstrap.  

\subsection{Estimating Implied $\Gamma$}

When using $\bskap_{1+}$, our estimated shrinkage factor will be given by 
%\[ \lambda_{1+} = 1-\left(1  - \left(1 - \frac{2\left( \htauo - \htaur\right)^\tran \bsD^2 \bsSig_r\left( \htauo - \htaur\right) }{\Tr(\bsSig_r \bsD) \left( \left( \htauo - \htaur\right)^\tran \bsD\left( \htauo - \htaur\right)\right)}\right)\frac{\Tr(\bsSig_r \bsD)}{\left( \htauo - \htaur\right)^\tran \bsD\left( \htauo - \htaur\right) } \right)_+\,.\] 
\[ \lambda_{1+} = 1-\left(1  - \frac{a_1^{\star}\Tr(\bsSig_r \bsD)}{\left( \htauo - \htaur\right)^\tran \bsD\left( \htauo - \htaur\right) } \right)_+\,.\] 

Once we posit a value for $\Gamma$, Zhao's method can be used to estimate the worst-case bias and variance when using SIPW to estimate the causal effect in the observational study within each stratum $k$. For each stratum $k$, we solve for the extrema -- the largest and smallest possible estimates that are consistent with the observational data from that stratum and the sensitivity model -- via the proposed linear fractional program. Subtracting out the point estimate of the causal effect and taking absolute values, this will give us two possible estimates of the bias of $\hat \tau_{ok}$: $\widehat{\text{Bias}}_{lk}(\Gamma)$, derived from the lower bound; and $\widehat{\text{Bias}}_{rk}(\Gamma)$, derived from the upper bound. 

We then draw repeated bootstrap replicates from the observational units within each stratum, and compute the extrema within each replicate. Variance is estimated by computing the variance across the replicates for each of the upper and lower bounds. We obtain estimates $\widehat{\text{Var}}_{lk}(\Gamma)$ and $\widehat{\text{Var}}_{rk}(\Gamma)$. Finally, we can choose
\begin{align*}
d_i^2 \widehat{\text{Bias}}_{k}(\Gamma) + d_i \widehat{\text{Var}_{k}}(\Gamma) &= \max \left(  d_i^2 \widehat{\text{Bias}}_{lk}(\Gamma) + d_i \widehat{\text{Var}_{lk}}(\Gamma),\right.\\&\left.  d_i^2 \widehat{\text{Bias}}_{rk}(\Gamma) + d_i \widehat{\text{Var}_{rk}}(\Gamma)\right) \,.
\end{align*}
These values can be directly plugged into the definition of $\lambda_{\text{opt}}$ to obtain an estimate $\lambda(\Gamma)$. 

This points to a simple algorithm for estimating the ``implied $\Gamma$" of our shrinkage estimate when using $\lambda_{1+}$, assuming that $\Gamma$ is shared across strata $k$: 
\begin{itemize}
\item Obtain $\lambda_{1+}$
\item Perform a binary search of $\Gamma$ values until $|\lambda_{1+} - \lambda(\Gamma)| < \epsilon$ for some small choice of $\epsilon$. Denote this value $\Gamma_{\text{imp}}$. 
\end{itemize}

There will be some randomness to the algorithm due to the bootstrap estimation of the variance, but with a sufficiently large number of replicates and reasonable choice of the tolerance parameter $\epsilon$, the algorithm should quickly converge. The resulting value $\Gamma_{\text{imp}}$ provides an interpretable notion of the bias for the analyst. If $\Gamma_{\text{imp}}$ lies within a range that matches the analyst's intuition, this provides license to proceed with the analysis.  

If it is unreasonably small -- say, $\Gamma_{\text{imp}} = 1.01$ -- then this signifies that $\bskap_{1+}$ is relying more heavily on $\htauo$ than the analyst thinks is reasonable. In such a case, the analyst has several options. She can simply use $\htaur$ and ignore the observational data. She can also essentially ``reverse" the process given above by choosing a value of $\Gamma$ that she considers reasonable, computing $\lambda(\Gamma)$, and then using the estimator
\[ \lambda(\Gamma) \htaur + (1 - \lambda(\Gamma)) \htauo \] 
to estimate the causal effects of interest. 

\section{Simulations}\label{sec:dataAnalysis}

\subsection{Simulation Set-Up}

We demonstrate the risk reduction for our proposed estimators under a variety of simulated scenarios. Our settings are partially patterned on those used in \cite{rosenman2018propensity}. 

In all of our simulations, our covariates $\bsX_i\in\real^3$ for $i \in \mathcal{O} \cup \mathcal{R}$. The observational study has $n_o=10{,}000$ subjects while the RCT has $n_r = 1{,}000$. On each new sampling of the covariates, we first sample a covariance matrix $\bsSig \in \mathbb{R}^{3 \times 3}$, such that each covariate has unit variance, and covariances are randomly 0 with $1/2$ probability, and $\pm 0.1$ with $1/4$ probability. This structure was used in \cite{rosenman2018propensity} because it is roughly consistent with the covariance structure present in the data from the Women's Health Initiative \citep{writing2002risks}. We then generate $\bsX_i\simiid\dnorm(0,\Sigma)$ for $i\in\odb \cup \rct$ and generate $U_i = 1/3 \cdot \mathbbm{1}^T \bsX_i + \eta _i$ where $\eta_i \sim \dnorm(0, 1/4)$. Hence, the unmeasured covariate has some stochastic contribution but is also correlated with the measured covariates.  

For the control condition, outcomes are generated as
\[ Y_i(0) = \bsX_i^\tran\boldsymbol{\beta} + U_i +  \err_i,\quad\text{for $\boldsymbol{\beta} = (1,1,1)^\tran$} \] 
for $i \in \mathcal{O} \cup \mathcal{R}$. The $\err_i$ are generated as IID $\dnorm(0,1)$ random variables.

We assume the treatment effect varies solely as a function of the second covariate (i.e. the second column of $\bsX_i)$, which we refer to as $\bsX_{i2}$. We will simulate under $K = 6$ (``few strata") and $K = 20$ (``many strata") conditions. We will also simulate with both equal- and variable-sized strata in order to include cases in which $\htaur$ is approximately homoscedastic versus heteroscedastic. In the heteroscedastic case, we suppose the first half of the strata contain approximately $2/(3K)$ units per stratum while the latter half contain approximately $4/(3K)$. The strata are defined based on the associated quantiles of the normal distribution, since we know $\bsX_{i2}$ will follow an approximately normal distribution for $i \in \mathcal{O}$ and $i \in \mathcal{R}$. We draw the treatment effects for each stratum according to a $\text{Uniform}(0, 1)$ distribution, and then linearly scale the effects in order to ensure that the Cohen's $d$ coefficient \citep{cohen1988statistical} equals 0.2 in the observational study. This corresponds to what Cohen calls a ``weak" effect. 

We sample the covariates and potential outcomes 25 times; for each choice, we sample the treatment assignments 20 times, for a total of 500 simulations. The treatment variables in the observational study are sampled as independent Bernoulli random variables with
\[p_i = \Pr(W_i = 1) = \frac{1}{1 + e^{-\boldsymbol{\gamma}^\tran\bsX_i + U_i}}.\] 
In order to induce a high level of selection bias in the observational study, we choose $\boldsymbol{\beta} = \boldsymbol{\gamma}$. Note that, because $p_i$ also depends on $U_i$, we cannot fully account for the selection bias by making use of inverse probability of treatment weighting. For units in the RCT, we randomly select half of the units within each stratum and assign them the treatment. 

We consider the performance of competing estimators. We include four of our proposed estimators, $\bskap_{1+}, \bskap_{1+}^{\star}, \bskap_{2+}$, and $\bskap_{2+}^{\star}$. We also consider Green and Strawderman's estimators, $\bsdelt_1$ and $\bsdelt_2$. Lastly, we compute an oracle estimator, which takes a convex combination of $\htaur$ and $\htauo$ weighted by the true optimal inverse-MSE weight, $\lambda_{\text{opt}}$. For each estimator, we estimate the risk via the average loss over the 500 simulations. We use $d_k = n_{ok}/n_o$ weighting scheme discussed earlier, which will yield $\bsD \approx 1/K \cdot \ident_K$ in the homoscedastic case but variable weights in the heteroscedastic case. Our main performance metric is the percent reduction in risk for these estimators as compared to the RCT-alone estimator $\htaur$. Larger risk reductions are preferred. Our estimators also outperform $\htauo$ in terms of risk across all the simulation conditions, though we do not directly report this risk reduction in the plots to follow. 

\subsection{Identical Observational and RCT Covariate Distributions}

We first consider the case in which the observational and RCT covariates are sampled from the same distribution, as described in the prior section. This is a somewhat ideal case, since it yields greater comparability between the datasets. To begin, we suppose that no effort is made to account for the (considerable) selection bias in the observational study. Results from these simulations are given in Figure \ref{fig:sim_Ideal_NoIPTW}.

We see immediately that all the estimators yield an improvement relative to use $\htaur$. The improvements are typically quite modest -- on the order of 1\% for our estimators in the 6-strata case, and 5\% in the 20-strata cases. This owes directly to the high bias in $\htauo$. For context, the risk of $\htauo$ was about 50 times that of the risk of $\htaur$ in the 6-strata simulations, and about 15 times larger in the 20-strata simulations. Hence, it is challenging to make use of these highly biased data to improve estimation, but we are still able to realize risk reductions.  

In the case of 6 strata, the best performing estimator is Green and Strawderman's $\bsdelt_1$ in the similar-size strata condition, while $\bskap_{1+}^{\star}$ does best in the variable-size condition. For 20 strata, $\bskap_{1+}$ is the winner in the similar-size condition and $\bskap_{1+}^{\star}$ in the variable-size condition. In general with these data, we observe that estimators making use of a single shrinkage factor across components (those with a ``1" subscript) tend to outperform those with component-level shrinkage factors. However, the only true laggard across the four conditions appears to be $\bsdelt_2$. 

Note also that we are only able to realize risk reductions about half that of the oracle when there are few strata. As dictated by the theory, we are able to get somewhat closer to oracle performance when there are more strata. 

%By comparison, in the 20-strata case, the estimators that make use of a single shrinkage factor -- $\bskap_1$ and $\bsdelt_1$ -- outperform their counterpart estimators, which induce different shrinkage factors across components. Our estimator, $\bskap_1$, performs slightly better in the variable-size case while the Green and Strawderman estimator, $\bsdelt_1$, performs slightly better in the similar-size case. As dictated by the theory, we are able to get somewhat closer to oracle performance when there are more strata. 

\begin{figure}
\centering
\includegraphics[scale = 0.34]{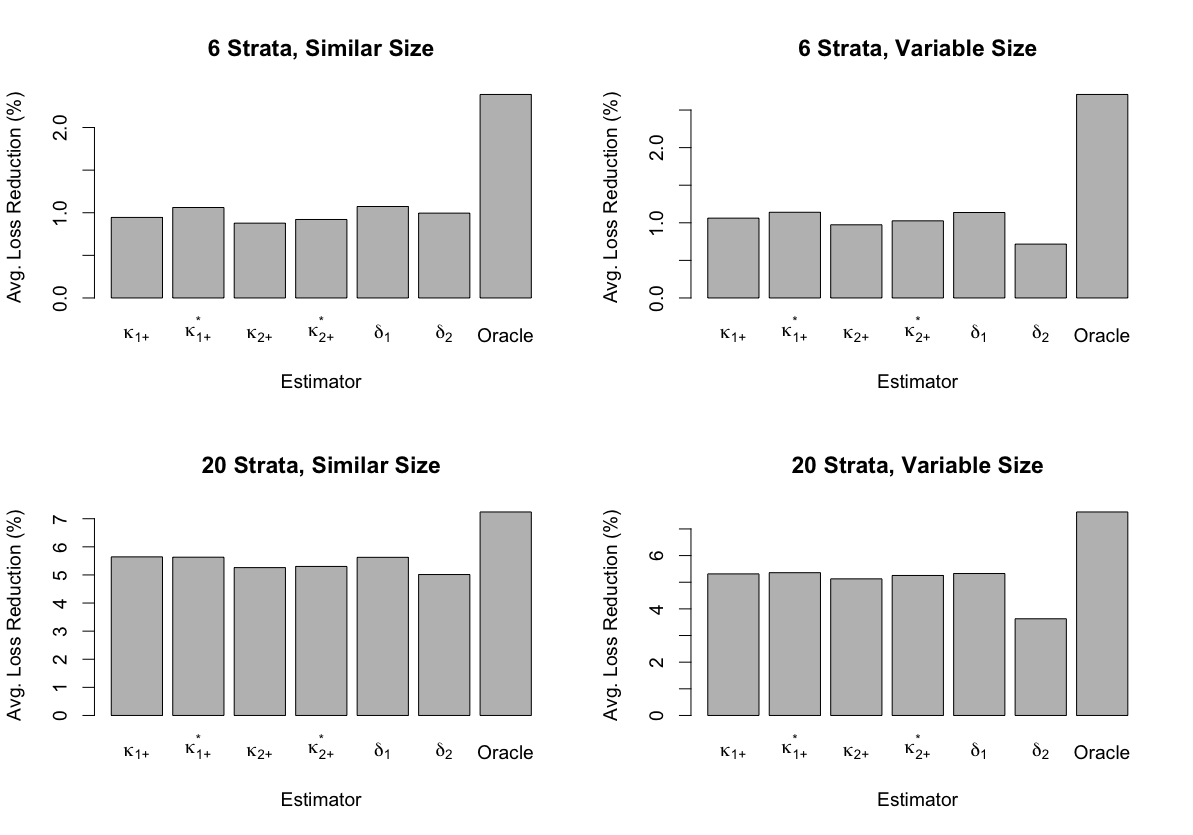}
\caption{\label{fig:sim_Ideal_NoIPTW}Percent reduction in risk relative to $\htaur$ for our proposed estimators, the Green and Strawderman estimators, and an oracle under four different conditions. Here, we assume $\htauo$ is computed without any adjustment for selection bias, yielding a highly biased estimator.}
\end{figure}

Much more substantial risk reductions are possible if we are able to reduce the bias of $\htauo$. Hence, we compute the same simulations but alter the estimation strategy in the observational dataset by using stabilized inverse probability of treatment weighting in each stratum, as described in Section \ref{subsubsec:psa}. Though we cannot remove all of the bias due to the influence of the unmeasured confounder, we can remove a large portion. The results are given in Figure \ref{fig:sim_Ideal_IPTW}. 

\begin{figure}
\centering
\includegraphics[scale = 0.34]{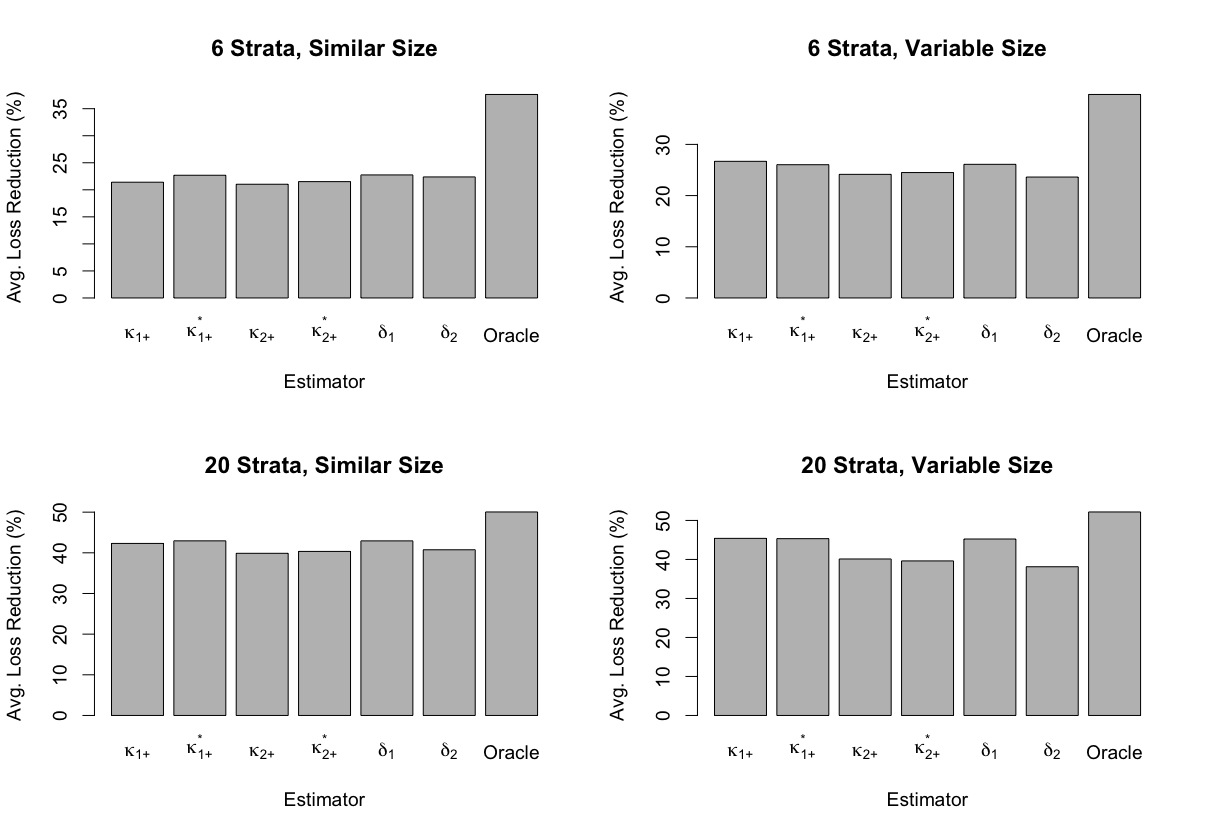}
\caption{\label{fig:sim_Ideal_IPTW} Percent reduction in risk relative to $\htaur$ for our proposed estimators, the Green and Strawderman estimators, and an oracle under four different conditions. Here, we assume $\htauo$ is computed by stabilized inverse probability of treatment weighting, such that some of the selection bias is removed.}
\end{figure}

First, we note that the risk reductions are much larger in magnitude -- on the order of 25\% for our estimators in the 6-strata case, and 40\% in the 20-strata case. This owes directly to the bias reduction in $\htauo$, whose risk is only approximately 40\% higher than that of $\htaur$ in the 6-strata case and almost identical to that of $\htaur$ in the 20-strata case. 

Green and Strawderman's $\bsdelt_1$ does best in the 6-strata, similar-size condition, while our estimators do best in the remaining scenarios: $\bskap_{1+}^{\star}$ does best in the 20-strata, similar-size condition and $\bskap_{1+}$ does best in the two variable-size conditions. Again, we see that when there are more strata, estimators using a single shrinkage factor tend to do better their counterparts using distinctive shrinkage factors across components. Moreover, we are able to get quite close to oracle performance, especially when using 20 strata. 

\subsection{Differing Observational and RCT Covariate Distributions}

We also consider the case where the distributions of the covariates differ between the observational and experimental studies. To induce the discrepancy, we first sample the mean vector for the observational covariates within each of the outer simulation loops, where each of the three entries is drawn from a $\text{Uniform}(-1/2, 1/2)$ distribution. This yields a mean vector $\mu_o$, and we then sample $\bsX_i\simiid\dnorm(\mu_o,\Sigma)$ for $i \in \mathcal{O}$. The RCT covariates are still sampled as $\bsX_i \simiid \dnorm(0, \Sigma)$. The strata are again defined by $\bsX_{i2}$, with the same quantiles used as in the prior case. In practice, this means we no longer have $\bsSig_o \approx \bsSig_r$ in any of the conditions. 

Results without stabilized IPW adjustment for $\htauo$ are given in Figure \ref{fig:sim_Differential_NoIPTW}. In general, all estimators have degraded somewhat in performance relative to the oracle. The Green and Strawderman estimators tend to do better in this setting: $\bsdelt_1$ outperforms in the first three conditions, while $\bskap_{2+}^{\star}$ does best in the 20-strata, variable-size condition. 

%, and our variable-shrinkage estimators, $\bskap_2$ and $\bsdelt_2$, lag substantially behind their counterpart estimators. %$\bskap_2$ underperforms in all of these simulations, and actually yield a risk increase in the 6-strata, variable-size case. Fortunately, we would know this ahead of time; in the majority of the data draws, we find $4 \max_k d_k^2 \sigma_{rk}^4 > \sum_k d_k^2 \sigma_{rk}^4$, so Lemma \ref{lemma:delta2DominatesTauR} does not apply. 

%While $\bskap_1$ is competitive in all four scenarios, it is the Green and Strawderman estimator, $\bsdelt_1$, that achieves the largest reductions in all conditions. The estimators appears to be slightly more robust in the case of a high-bias $\htauo$ and distributional differences between the observational and experimental covariates.  

\begin{figure}
\centering
\includegraphics[scale = 0.34]{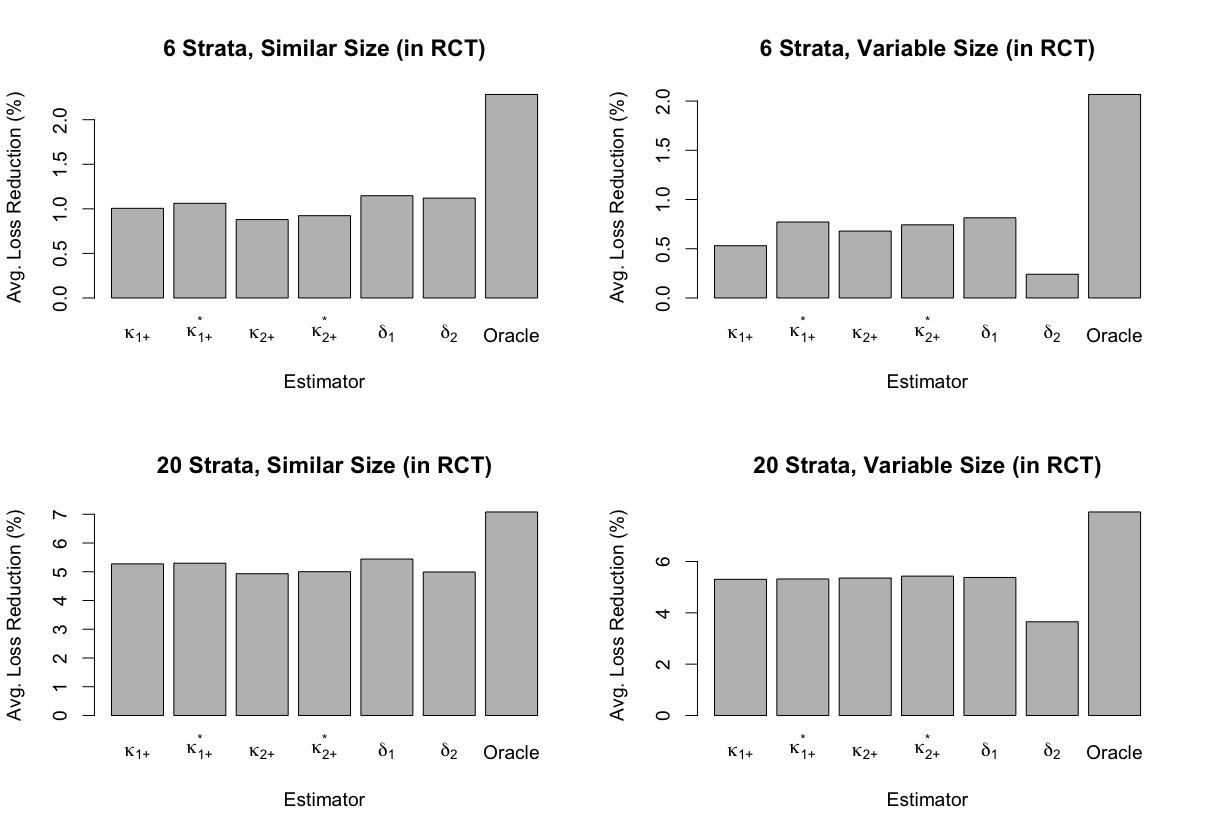}
\caption{\label{fig:sim_Differential_NoIPTW} Percent reduction in risk relative to $\htaur$ for our proposed estimators, the Green and Strawderman estimators, and an oracle under four different conditions. Here, we assume $\htauo$ is computed without any adjustment for selection bias, yielding a highly biased estimator. We also induce different distributions for the covariates $\bsX_i$ among the observational and RCT units.}
\end{figure}

Lastly, we recompute the estimators with stabilized IPW estimation used to compute $\htauo$. The results are given in Figure \ref{fig:sim_Differential_IPTW}. We again have very similar performance between $\bskap_{1+}^{\star}$ and $\bsdelt_1$, with $\bsdelt_1$ modestly edging $\bskap_{1+}^{\star}$ for the lead in each condition. 

%Surprisingly, $\bskap_2$ is our best performer in the case of 6 similar-size strata. For the remaining condition, we again have very similar performance between $\bskap_1$ and $\bsdelt_1$, with $\bsdelt_1$ modestly edging $\bskap_1$ for the lead in each condition. 

\begin{figure}
\centering
\includegraphics[scale = 0.34]{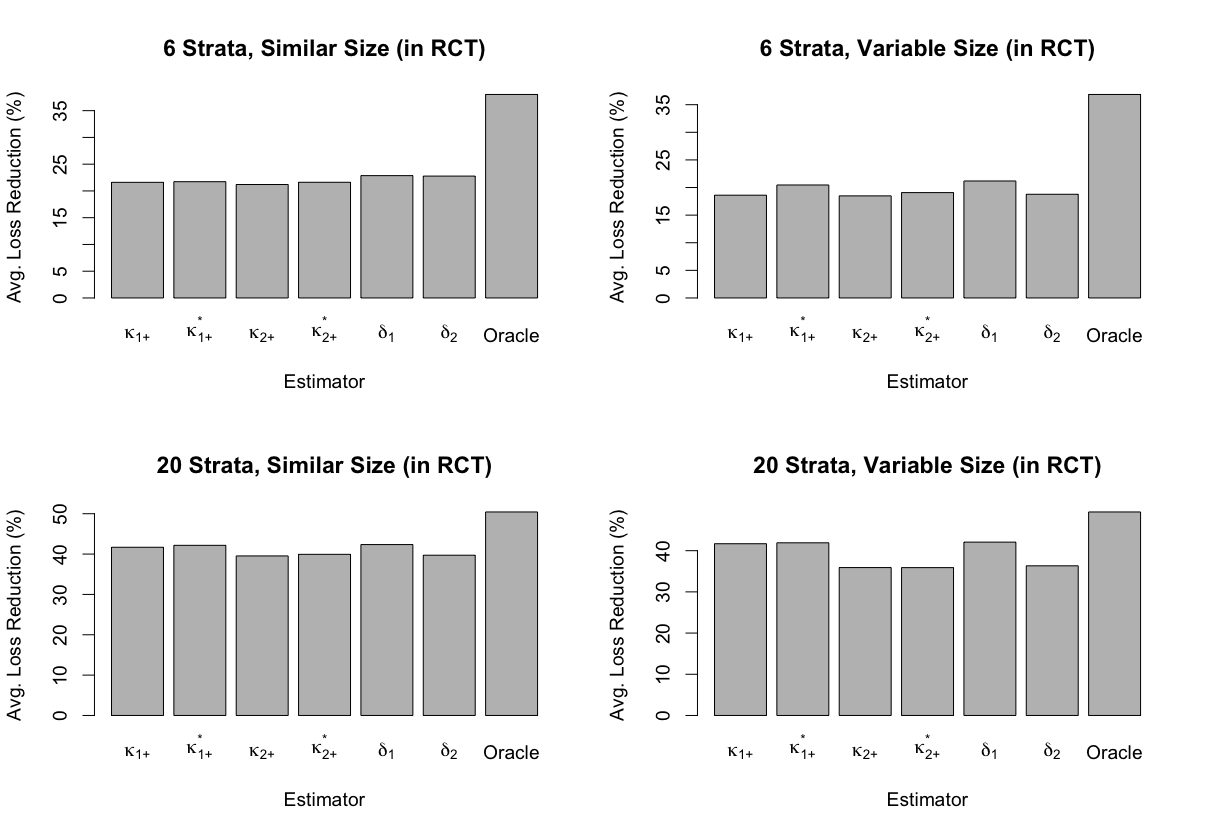}
\caption{\label{fig:sim_Differential_IPTW} Percent reduction in risk relative to $\htaur$ for our proposed estimators, the Green and Strawderman estimators, and an oracle under four different conditions. Here, we assume $\htauo$ is computed by stabilized inverse probability of treatment weighting, such that some of the selection bias is removed. We also induce different distributions for the covariates $\bsX_i$ among the observational and RCT units.}
\end{figure}

\section{Conclusion}\label{sec:conclusions}

There exists a considerable history in the statistics literature of minimizing unbiased risk estimates to obtain tuning parameters. Drawing on this work, we have developed a procedure for deriving shrinkage estimators that trade off between a biased and unbiased estimator of a $K$-dimensional parameter. We first generalize a result from Strawderman in order to obtain an unbiased risk estimate in our setting. We then posit a structure for the tradeoff estimator; derive its functional form by minimizing the unbiased risk estimate; and optionally further optimize to address the penalty we incur for using the data to both estimate the shrinkage factor and the estimand itself. We used this procedure to introduce new estimators, termed $\bskap_{1+}$ and $\bskap_{2+}$, and their analogues, $\bskap_{1+}^{\star}, \bskap_{2+}^{\star}$, for which we proved testable finite-$K$ conditions under which they have lower risk than the unbiased estimator. We also showed that both estimators achieve a notion of optimality in the limit of infinite $K$. 

We are interested in deploying these methods to address a problem in causal inference: how to combine observational and experimental data to estimate causal effects. Observational data is ubiquitous, but because treatment is not randomly assigned, the causal estimates it yields are biased. Propensity score methods can be used to reduce this bias. But it cannot be eliminated unless we are willing to make the burdensome assumption that all confounding variables are measured. By contrast, experimental data yields causal estimates that are unbiased, but often have higher variance, because such data is typically expensive to obtain. Our estimators provide a template for combining these two types of data in order to manage the bias-variance tradeoff and yield lower overall risk.  

For the practical use of these estimators, we develop connections to a sensitivity analysis method proposed by Zhao and co-authors. We also explore the estimators' performance on a simulated dataset in which we seek to quantify weak causal effects using a small experiment and a larger observational study suffering from unmeasured confounding. In the simulations, we compare the performance of our estimators against two estimators proposed by Green and Strawderman, $\bsdelt_1$ and $\bsdelt_2$. We find that all our estimators do better than an estimator derived solely from the experimental data, under all tested conditions. $\bskap_1$ typically achieves lower risk than Green and Strawderman's estimators when the covariate distribution is similar in the observational and experimental data. However, $\bsdelt_1$ is slightly more robust when the distributions differ. 

There are numerous potential extensions to this work. We have explored two shrinkage structures in this text -- shrinkage by a constant factor, and shrinkage by a variance-weighted factor -- but our procedure is general and can be used to derive alternative estimators. We might, for example, incorporate auxiliary information in order to guess which strata suffer from the most bias in the observational study. We could then design our estimator to shrink less toward the observational estimate in these strata. Or, we might be interested in a thresholding approach in which we rely solely on $\hat \tau_{rk}$ in strata $k$ for which $\hat \tau_{rk}$ differs more than some threshold $\Delta$ from $\hat \tau_{ok}$. 

In future work, we will also seek to weaken our assumptions. To establish theoretical results, we have supposed that the strata are well-defined in the observational and experimental datasets, and that average treatment effects are shared within strata across the datasets. In many practical examples, analysts will be interested in heterogeneous treatment effect estimation. They will face a tradeoff between trying to estimate many subgroup treatment effects less accurately versus more stable estimation of fewer effects. Our methods require a minimum of four strata to guarantee a risk reduction, and potentially many more in a case with substantial heteroscedasticity or highly differential weights by stratum. The assumption of shared average treatment effects is also unlikely to hold exactly in practice. We will aim to develop practical guidelines for determining a reasonable set of strata for which to estimate causal effects, and measures of robustness to violations of Assumption \ref{ass:cte}. 

\section*{Acknowledgments}
Evan Rosenman was supported by Google, and by the Department of Defense (DoD) through
the National Defense Science \& Engineering Graduate Fellowship (NDSEG)
Program. This work was also supported by the NSF under grants DMS-1521145, DMS-1407397, and IIS-1837931.

\bibliographystyle{apalike} 
\bibliography{biblio}

\section*{Appendix}
\section{Proof of Theorem \ref{thm:delta1asymptotics}}\label{sec:proofOfDelta1Asymptotics}

The proof proceeds in several stages. We replicate the steps in a similar proof offered in \cite{xie2012sure} to prove that $ \bskap(\lambda_1^{\URE}, \htaur, \htauo)$ is asymptotically optimal. 

\begin{lemma}\label{lemma:uniflossapprox}
Assuming that
\begin{align}\label{conditions}
\limsup_{K \to \infty} \frac{1}{K} \sum_k d_k^2 \sigma_{rk}^2\xi_k^2&< \infty\, \\
\limsup_{K \to \infty} \frac{1}{K} \sum_k d_k^2 \sigma_{rk}^2 \sigma_{ok}^2 &< \infty\, \text{and}\\
\limsup_{K \to \infty} \frac{1}{K} \sum_k d_k^2 \sigma_{rk}^4 &<\infty
\end{align}
we have 
\[ \sup_{0 \leq \lambda \leq 1} \left| \URE(\bstau, \bskap(\lambda, \htaur, \htauo)) - \mathcal{L}(\bstau, \bskap(\lambda, \htaur, \htauo))  \right| \to 0\] 
in $L_2$ and in probability as the dimension $K \to \infty$. 
\end{lemma}
\begin{proof}
We show only the $L_2$ convergence, which implies convergence in probability. Assuming $\lambda$ lies in the unit interval, we can write
\begin{align*}
URE(\bstau, \bskap(\lambda, \htaur, \htauo)) - \mathcal{L}(\bstau, \bskap(\lambda, \htaur, \htauo)) &= \frac{1}{K} \left( \sum_k d_k \left( \sigma_{rk}^2 - \left( \hat \tau_{rk} - \tau_k \right)^2 \right) + \right. \\ &2 \left. \lambda \sum_k d_k \left( \sigma_{rk}^2 - \left( \hat \tau_{rk} - \tau_k \right) \left(\hat \tau_{rk} - \hat \tau_{ok} \right) \right)\right)
\end{align*}
and hence 
\begin{align*}
\sup_{0 \leq \lambda \leq 1} \left| \URE(\bstau, \bskap(\lambda, \htaur, \htauo)) - \mathcal{L}(\bstau, \bskap(\lambda, \htaur, \htauo))\right| & \leq  \left| \frac{1}{K} \sum_k d_k \left( \sigma_{rk}^2 - \left( \hat \tau_{rk} - \tau_k \right)^2 \right) \right| -\\ & \hspace{5mm} \sup_{0 \leq \lambda \leq 1} \left| \frac{2}{K} \lambda \sum_k d_k \left( \sigma_{rk}^2 - \left( \hat \tau_{rk} - \tau_k \right) \left(\hat \tau_{rk} - \hat \tau_{ok} \right)\right)\right| 
\end{align*}
We can consider the terms separately. For the first term, we observe
\begin{align*}
E \left( \left( \frac{1}{K} \sum_k d_k \left(\sigma_{rk}^2 - \left( \hat \tau_{rk} - \tau_k \right)^2 \right)\right)^2 \right) &= \frac{1}{K^2} \sum_k d_k^2 E \left( \left( \sigma_{rk}^2 - (\hat \tau_{rk} - \tau_k)\right)^2 \right) \\
&= \frac{1}{K^2} \sum_k d_k^2 \left( \sigma_{rk}^4 - 2 \sigma_{rk}^4 + 3 \sigma_{rk}^4 \right)\\
&= \frac{2}{K^2} \sum_k d_k^2 \sigma_{rk}^4 \to 0 
\end{align*} 
by our third regularity condition. 

For the second term, observe that, in general, 
\small
\begin{align*}
\sup_{0 \leq \lambda \leq 1} \left| \frac{2}{K} \lambda \sum_k d_k \left( \sigma_{rk}^2 - \left( \hat \tau_{rk} - \tau_k \right) \left(\hat \tau_{rk} - \hat \tau_{ok} \right)\right) \right| &\leq \sup_{1 \geq c_1 \geq \dots \geq c_K \geq 0}\frac{2}{K} \left| \sum_k c_k d_k \left( \sigma_{rk}^2 - \left( \hat \tau_{rk} - \tau_k \right) \left(\hat \tau_{rk} - \hat \tau_{ok} \right)\right) \right|  \,.
\end{align*}
\normalsize
Applying Lemma 2.1 from \cite{li1986asymptotic}, we observe 
\small
\[ \sup_{1 \geq c_1 \geq \dots \geq c_K \geq 0}\frac{2}{K} \left| \sum_k c_k d_k \left( \sigma_{rk}^2 - \left( \hat \tau_{rk} - \tau_k \right) \left(\hat \tau_{rk} - \hat \tau_{ok} \right)\right) \right| = \max_{1 \leq j \leq K} \frac{2}{K}\left|  \sum_{k = 1}^j d_k \left(  \sigma_{rk}^2 - \left( \hat \tau_{rk} - \tau_k \right) \left(\hat \tau_{rk} - \hat \tau_{ok} \right) \right) \right| \,. \] 
\normalsize
Observe that, for each value of $k$,
\begin{align*}
E \left(d_k \left( \sigma_{rk}^2 - \left( \hat \tau_{rk} - \tau_k \right) \left(\hat \tau_{rk} - \hat \tau_{ok} \right) \right)\right) &=  d_k \left(\sigma_{rk}^2 - E \left( \hat \tau_{rk}^2 \right) + \tau_k E \left( \hat \tau_{rk} \right) + E\left( \hat \tau_{ok} \hat \tau_{rk} \right) - \tau_k E \left( \hat  \tau_{ok} \right)\right) \\
&= 0\,,
\end{align*}
and thus for $M_j = \sum_{k  = 1}^j  d_k \left( \sigma_{rk}^2 - \left( \hat \tau_{rk} - \tau_k \right) \left(\hat \tau_{rk} - \hat \tau_{ok} \right)\right)$, $\{M_j: j = 1, 2, \dots\}$ forms a martingale. We can then use the $L^p$ maximal inequality to observe
\[ E \left( \max_{1 \leq j \leq K} M_j^2 \right) \leq 4 E \left(M_K^2\right) = 4 \sum_k d_k^2 \left( \sigma_{rk}^2\xi_k^2 + \sigma_{rk}^2 \sigma_{ok}^2 + 2 \sigma_{rk}^4\right) \,.\] 
Our regularity conditions thus guarantee that 
\[ E \left(\max_j \left( \frac{2}{K} M_j^2\right) \right) \to 0 \] 
which tells us 
\[ \sup_{0 \leq \lambda \leq 1} \left| \frac{2}{K} \lambda \sum_k d_k \left( \sigma_{rk}^2 - \left( \hat \tau_{rk} - \tau_k \right) \left(\hat \tau_{rk} - \hat \tau_{ok} \right)\right)\right| \to 0\] 
in $L_2$ as $K \to \infty$. 
\end{proof}

Lemma \ref{lemma:uniflossapprox} tells us that our risk estimate $\URE(\bstau, \bskap(\lambda, \htaur, \htauo))$ is close to the actual loss of our estimator $ \mathcal{L}(\bstau, \bskap(\lambda, \htaur, \htauo))$ as the dimension grows large. It follows that minimizing our risk estimate should yield a competitive estimator. We can formalize this result by considering the risk of any other estimator with a constant shrinkage factor, $\bskap(\lambda, \htaur, \htauo)$.  

\begin{lemma}\label{lemma:bestEstimator}
Assuming Conditions 1-3, we have
\[ \lim_{K \to \infty} \left( R(\bstau, \bskap(\lambda_1^{\URE}, \htaur, \htauo))-R(\bstau, \bskap(\lambda, \htaur, \htauo))\right) \leq 0 \,.\]
for any choice of $\lambda$. 
\end{lemma}
\begin{proof}
Observe 
\begin{align*}
\mathcal{L}(\bstau, \bskap(\lambda_1^{\URE}, \htaur, \htauo))-\mathcal{L}(\bstau, \bskap(\lambda, \htaur, \htauo)) &= \left( \mathcal{L}(\bstau, \bskap(\lambda_1^{\URE}, \htaur, \htauo)) - \URE(\bstau, \bskap(\lambda_1^{\URE}, \htaur, \htauo))\right) + \\
& \left( \URE(\bstau, \bskap(\lambda_1^{\URE}, \htaur, \htauo)) - \URE(\bstau, \bskap(\lambda, \htaur, \tau_o))\right) + \\
& \left(  \URE(\bstau, \bskap(\lambda, \htaur, \htauo)) - \mathcal{L}(\bstau, \bskap(\lambda, \htaur, \htauo)) \right) 
\end{align*}
The second term must be negative because $\bskap(\lambda_1^{\URE}, \htaur, \htauo)$ minimizes the unbiased risk estimate among all choices of $\lambda$. Hence, we have
\[ \mathcal{L}(\bstau, \bskap(\lambda_1^{\URE}, \htaur, \htauo))-\mathcal{L}(\bstau, \bskap(\lambda, \htaur, \htauo)) \leq 2 \sup_{0 \leq \lambda' \leq 1} \left| \URE(\bstau, \bskap(\lambda', \htaur, \htauo)) - \mathcal{L}(\bstau, \bskap(\lambda', \htaur, \htauo)) \right| \] 
Taking expectations of both sides yields
\[ \left( R(\bstau, \bskap(\lambda_1^{\URE}, \htaur, \htauo))-R(\bstau, \bskap(\lambda, \htaur, \htauo))\right) \leq 2 E \left( \sup_{0 \leq \lambda' \leq 1} \left| \URE(\bstau, \bskap(\lambda', \htaur, \htauo)) - \mathcal{L}(\bstau, \bskap(\lambda', \htaur, \htauo)) \right|\right)\] 

From Lemma \ref{lemma:uniflossapprox} we know that the term on right hand side goes to 0 in $L_2$ (and thus in $L_1$) as $K \to \infty$. Hence
\[ \lim_{K \to \infty} \left( R(\bstau, \bskap(\lambda_1^{\URE}, \htaur, \htauo))-R(\bstau, \bskap(\lambda, \htaur, \htauo))\right) \leq 0 \] 
as desired. 
\end{proof}

\section{Proof of Theorem \ref{thm:delta2asymptotics}}\label{sec:proofOfDelta2Asymptotics}

The proof is substantively similar to that of Theorem \ref{sec:proofOfDelta1Asymptotics}. 

\begin{lemma}\label{lemma:uniflossapprox2}
Assuming that
\begin{align}\label{conditions}
\limsup_{K \to \infty} \frac{1}{K} \sum_k d_k^2 \sigma_{rk}^6\xi_k^2&< \infty\, \\
\limsup_{K \to \infty} \frac{1}{K} \sum_k d_k^2 \sigma_{rk}^6 \sigma_{ok}^2 &< \infty\, \\
\limsup_{K \to \infty} \frac{1}{K} \sum_k d_k^2 \sigma_{rk}^8 &<\infty \, \text{and}\\
\limsup_{K \to \infty} \frac{1}{K} \sum_k d_k^2 \sigma_{rk}^4 &<\infty
\end{align}
we have 
\[ \sup_{0 \leq \lambda \leq 1} \left| \URE(\bstau, \bskap(\lambda \bsSig_r, \htaur, \htauo)) - \mathcal{L}(\bstau, \bskap(\lambda \bsSig_r, \htaur, \htauo))  \right| \to 0\] 
in $L_2$ and in probability as the dimension $K \to \infty$. 
\end{lemma}
\begin{proof}
Again, we show the $L_2$ convergence. Analogous computations to those in the proof of Lemma \ref{lemma:uniflossapprox} show 
\begin{align*}
\sup_{0 \leq \lambda \leq 1} \left| \URE(\bstau, \bskap(\lambda \bsSig_r, \htaur, \htauo)) - \mathcal{L}(\bstau, \bskap(\lambda \bsSig_r, \htaur, \htauo))\right| & \leq  \left| \frac{1}{K} \sum_k d_k \left( \sigma_{rk}^2 - \left( \hat \tau_{rk} - \tau_k \right)^2 \right) \right| +\\ & \hspace{-27mm} \sup_{0 \leq \lambda \leq 1} \left| \frac{2}{K} \lambda \sum_k d_k \sigma_{rk}^2 \left( \sigma_{rk}^2 - \left( \hat \tau_{rk} - \tau_k \right) \left(\hat \tau_{rk} - \hat \tau_{ok} \right)\right)\right| 
\end{align*}
The first term is unchanged from the proof of Lemma \ref{lemma:uniflossapprox}, so we can rely on our final regularity condition to assert its convergence. For the second term, we can use analogous machinery to observe that 
\small
\begin{align*}
\sup_{0 \leq \lambda \leq 1} \left| \frac{2}{K} \lambda \sum_k d_k \sigma_{rk}^2 \left( \sigma_{rk}^2 - \left( \hat \tau_{rk} - \tau_k \right) \left(\hat \tau_{rk} - \hat \tau_{ok} \right)\right) \right| &\leq \max_{1 \leq j \leq K} \frac{2}{K}\left|  \sum_{k = 1}^j d_k\sigma_{rk}^2 \left(  \sigma_{rk}^2 - \left( \hat \tau_{rk} - \tau_k \right) \left(\hat \tau_{rk} - \hat \tau_{ok} \right) \right) \right| \,. 
\end{align*}
\normalsize
Observe that, for each value of $k$,
\begin{align*}
E \left(d_k \sigma_{rk}^2 \left( \sigma_{rk}^2 - \left( \hat \tau_{rk} - \tau_k \right) \left(\hat \tau_{rk} - \hat \tau_{ok} \right) \right)\right) &=  d_k \sigma_{rk}^2 \left(\sigma_{rk}^2 - E \left( \hat \tau_{rk}^2 \right) + \tau_k E \left( \hat \tau_{rk} \right) + E\left( \hat \tau_{ok} \hat \tau_{rk} \right) - \tau_k E \left( \hat  \tau_{ok} \right)\right) \\
&= 0\,,
\end{align*}
and thus for $\tilde M_j = \sum_{k  = 1}^j  d_k \sigma_{rk}^2 \left( \sigma_{rk}^2 - \left( \hat \tau_{rk} - \tau_k \right) \left(\hat \tau_{rk} - \hat \tau_{ok} \right)\right)$, $\{\tilde M_j: j = 1, 2, \dots\}$ forms a martingale. We can then use the $L^p$ maximal inequality to observe
\[ E \left( \max_{1 \leq j \leq K} \tilde M_j^2 \right) \leq 4 E \left(\tilde M_K^2\right) = 4 \sum_k d_k^2 \sigma_{rk}^4 \left( \sigma_{rk}^2\xi_k^2 + \sigma_{rk}^2 \sigma_{ok}^2 + 2 \sigma_{rk}^4\right) \,.\] 
Our first three regularity conditions thus guarantee that 
\[ E \left(\max_j \left( \frac{2}{K} \tilde M_j^2\right) \right) \to 0 \] 
which tells us 
\[ \sup_{0 \leq \lambda \leq 1} \left| \frac{2}{K} \lambda \sum_k d_k \sigma_{rk}^2 \left( \sigma_{rk}^2 - \left( \hat \tau_{rk} - \tau_k \right) \left(\hat \tau_{rk} - \hat \tau_{ok} \right)\right)\right| \to 0\] 
in $L_2$ as $K \to \infty$. 
\end{proof}

\end{document}